\documentclass[envcountsect,envcountsame]{llncs}
\usepackage{amssymb,stmaryrd}
\usepackage{amsmath}
\usepackage{epsfig}
\usepackage{makeidx}
\usepackage{enumerate}
\usepackage{algorithmic}
\usepackage[ruled,section]{algorithm}

\algsetup{indent=2em}
\renewcommand{\algorithmicendif}{\textbf{fi}}

\newcommand{\OLIF}[2]
{\STATE \algorithmicif \mbox{ #1 }\algorithmicthen\mbox{ #2 }\algorithmicendif}
\newcommand{\OLELIF}[3]
{\STATE \algorithmicif \mbox{ #1 }\algorithmicthen\mbox{ #2 }\
  \algorithmicelse \mbox{ #3} \algorithmicendif}
\newcommand{\rst}[1]{_{\upharpoonright{#1}}}
\newcommand{\is}{\leftarrow}

\begin{document}

\title{Multi-Stage Improved Route Planning Approach}
\subtitle{Theoretical Foundations}%
\author{Petr Hlin\v{e}n\'y ~\and~ Ondrej Mori\v{s}}

\institute{Faculty of Informatics, Masaryk University \\
  Botanick\'a 68a, 602 00 Brno, Czech Republic  \\
  \email{hlineny@fi.muni.cz, xmoris@fi.muni.cz}}

\maketitle

\begin{abstract}
A new approach to the static route planning problem, based on a multi-staging 
concept and a \emph{scope} notion, is presented. The main goal (besides implied
efficiency of planning) of our approach is to address---with a solid theoretical
foundation---the following two practically motivated aspects: a \emph{route 
comfort} and a very \emph{limited storage} space of a small navigation device,
which both do not seem to be among the chief objectives of many other studies.
We show how our novel idea can tackle both these seemingly unrelated aspects at
once, and may also contribute to other established route planning approaches 
with which ours can be naturally combined. We provide a theoretical proof that 
our approach efficiently computes exact optimal routes within this concept, 
as well as we demonstrate with experimental results on publicly available road 
networks of the US the good practical performance of the solution.
\end{abstract}

\section{Introduction}
\label{sec:introduction}

The single pair shortest path SPSP problem in world road networks, also 
known as the {\em route planning problem}, has received considerable attention 
in the past decades. However classical algorithms like Dijkstra's or A* are 
in fact ''optimal`` in a theore\-tical sense, they are not well suitable for 
huge graphs representing real-world road networks having up to tens millions 
of edges. In such situation even an algorithm with a linear time complexity 
cannot be feasibly run on a mobile device lacking computational power and 
working memory. 

What can be done better? We focus on the \emph{static route planning problem} 
where the road network is static in the sense that the underlying graph and 
its attributes do not change in time. Thus, a feasible solution lies in 
suitable \emph{prepro\-cessing} of the road network in order to improve both 
time and space complexity of subsequent \emph{queries} (to find an optimal 
route from one position to~another). However, to what extent such a 
preprocessing is {\em limited in the size} of the precomputed auxiliary data?
It is not hard to see that there is always some trade-off between this
storage space requirement and the efficiency of queries---obviously, one 
cannot precompute and store all the optimal routes in advance. See also a 
closer discussion below.

\paragraph{\bf Related Work. }

Classical techniques of the static route planning are represented by 
Dijkstra's algorithm \cite{Dijkstra1959}, A* algorithm \cite{Hart1968}
and their bidirectional variations \cite{Pohl1969}. In the last decade,
two sorts of more advanced techniques have emerged and become popular. 
The first one prunes the search of Dijkstra's or A* algorithms using 
preprocessed information This includes, in particular, reach-based 
\cite{Gutman2004,Goldberg2005A}, landmarks \cite{Goldberg2005B,Goldberg2005C},
combinations of those \cite{Goldberg2007}, and recent hub-based labeling 
\cite{Abraham2011}. 

The second sort of techniques (where our approach conceptually fits, too) 
exploits a road network structure with levels of hierarchy to which a route 
can be decomposed into. For instance, highway and contraction hierarchies 
\cite{Sanders2006,Schultes2007,Geiberger2008}, transit nodes \cite{Bast2007},
PCD \cite{Maue2009} and SHARC routing \cite{Bauer2010,Brunel2010} represent 
this sort. Still, there are also many other techniques and combinations,
but---due to lack of space---we just refer to Cherkassky et al. 
\cite{Cherkassky1994}, Delling et al. \cite{Delling2009A,Delling2009B}, 
and Schultes \cite{Schultes2008}. Finally, we would like to mention the 
interesting notion of highway dimension \cite{Abraham2010A}, and the ideas 
of customizable \cite{Abraham2011} and mobile \cite{Sanders2008} route 
planning. 

\paragraph{\bf Our Contribution. }

We summarize the essence of all our contribution already here, while we 
implicitly refer to the subsequent sections for precise definitions, 
algorithms, proofs, and further details.

First of all, we mention yet another integral point of practical route
planning implementations---human-mind intuitiveness and comfortability 
of the computed route. This is a rather subjective requirement which is not 
easy to formalize via mathematical language, and hence perhaps not often 
studied together with the simple precise ``shortest/fastest path'' 
utility function in the papers.

\begin{itemize}
\item[] \emph{Intuitiveness and comfort of a route}: Likely everyone using car 
  navigation systems has already experienced a situation in which the computed 
  route contained unsuitable segments, e.g.\ tricky shortcuts via low-category 
  roads in an unfamiliar area. Though such a shortcut might save a few seconds 
  on the route in theory, regular drivers would certainly not appreciate it and
  the real practical saving would be questionable. This should better be 
  avoided.
\end{itemize}

Nowadays, the full (usually commercially) available road network data contain 
plenty of additional metadata which allow it to detect such unreasonable routes.
Hence many practical routing implementations likely contain some kinds of
rather {\em heuristic penalization} schemes dealing with this comfortability issue.
We offer here a mathematically sound and precise formal solution to the
route comfort issue which builds on a new theoretical concept of \emph{scope} 
(Sec.~\ref{sec:scope}).

The core idea of a scope and of scope-admissible routes can be informally
outlined as follows: The elements (edges) of a road network are spread into 
several {\em scope levels}, each associated with a {\em scope value}, such that
an edge $e$ assigned scope value $s_e$ is {\em admissible} on a route $R$ if, 
before or after reaching~$e$, such $R$ travels distance less than $s_e$ on 
edges of higher scope level than~$e$. Intuitively, the scope levels and values 
describe suitability and/or importance of particular edges for long-range 
routing. This is in some sense similar to the better known concept of 
{\em reach} \cite{Gutman2004}; but in our case the importance of an edge is to 
be decided from available network metadata and hence its comfort and intuitive
suitability is reflected, making a fundamental difference from the reach 
concept.

The effect seen on {\em scope admissible} routes (Def.~\ref{def:scope},
\ref{def:xadmissible},\ref{def:stadmissible}) is that low-level roads are 
fine ``near'' the start or target positions (roughly until higher-level roads
become available), while only roads of the highest scope levels are admissible 
in the long middle sections of distant routing queries. This nicely corresponds
with human thinking of intuitive routes, where the driver is presumably 
familiar with neighborhoods of the start and the target of his route (or, such 
a place is locally accessible only via low-level roads anyway). On contrary, 
on a long drive the mentally demanding navigation through unknown rural roads
or complicated city streets usually brings no overall benefit, even if it were 
faster in theory. 

To achieve good practical usability, too, road network segments are assigned
scope levels (cf.~Table~\ref{tab:scope_example}) according to expectantly
available metadata of the network (such as road categories, but also road 
quality and other, e.g.\ statistical information). It is important that
already experiments with publicly available TIGER/Line 2009 \cite{Tiger2009} 
US road data, which have metadata of questionable quality, show that the 
restriction of admissible routes via scope has only a marginal statistical 
effect on shortest distances in the network.

\smallskip
Furthermore, a welcome added benefit of our categorization of roads into scope 
levels and subsequent scope admissibility restriction is the following.

\begin{itemize}
\item[] \emph{Storage space efficiency}:
  We suggest that simply allowing to store ``linearly sized'' precomputed 
  auxiliary data (which is the case of many studies) may be too much in 
  practice. Imagine that setting just a few attributes of a utility function 
  measuring route optimality results in an exponential number of possibilities 
  to which the device has to keep separate sets of preprocessed auxiliary 
  data. In such a case a stricter storage limits should be imposed.
\end{itemize}

In our approach, preprocessed data for routing queries have to deal only with 
the elements of the highest scope(s). This allows us to greatly reduce the 
amount of auxiliary precomputed information needed to answer queries quickly. 
We use (Sec.~\ref{sub:preprocessing}) a suitably adjusted fast vertex-separator 
approach which stores only those precomputed boundary-pair distances (in the 
cells) that are admissible on the highest scope level. This way we can shrink 
the {\em auxiliary data size to less than $1\%$} of the road network size 
(Table~\ref{tab:preprocessing}) which is a huge improvement. 
Recall that {\em vertex-separator} preprocessing produces a 
partition of the network graph into moderate-size cells (several thousands 
of edges, say) such that only (selected) distances between pairs of their 
boundary vertices are precomputed.

Not to forget, our subsequent routing query algorithm (Sec.~\ref{sec:query})
then answers quickly and exactly (not heuristically) an optimal route
among all scope admissible ones between the given positions. 
The way we cope with scope admissibility (among other aspects) in a 
{\em route planning query}, using the precomputed auxiliary data,
It is briefly summarized as follows. 

\begin{itemize}
\item[]
\emph{Multi-staging approach}:\,
The computation of an optimal route is split into~two
(or possibly more -- with finer network metadata and hierarchical separators) 
different stages. In the local -- {\em cellular}, stage, a modification of 
plain Dijkstra's algorithm is used to reach the cell boundaries in such a 
way that lower levels are no longer admissible. Then in the global -- 
{\em boundary}, stage, an optimal connection between the previously 
reached boundary vertices found on the (much smaller) boundary graph 
given by auxiliary data.
\end{itemize}

Notice that the cellular stage may possibly cross the boundaries of a few 
adjacent cells if the start or target is near to them, but practical 
experiments show that such a case is quite rare. After all, the domain 
of a cellular stage is a small local neighbourhood, and Dijkstra's search 
can thus be {very fast} on it with additional help of a reach-like 
paramet\-rization (Def.~\ref{def:Sreach}). Then, handling 
the precomputed boundary graph in the global stage is very flexible---%
since no side restrictions exist there---and can be combined with virtually 
any other established route planning algorithm (see Sec.~\ref{sec:query}).
The important advantage is that the boundary graph is now much smaller
(recall, $<1\%$ in experiments) than the original network size, and hence 
computing on it is not only faster, but also more {working-memory 
efficient} which counts as well for mobile navigation devices.

\paragraph{\bf Paper organization. }
After the informal outline of new contributions, this paper continues with 
the relevant formal definitions---Section~\ref{sec:preliminaries} for route 
planning basics, and Section~\ref{sec:scope} for thorough description of the 
new scope admissibility concept. An adaptation of Dijkstra for scope is 
sketched in Sec.~\ref{sub:Sdijkstra}. Then Section~\ref{sec:algorithm} shows 
further details of the road network preprocessing (\ref{sub:preprocessing}) 
and query (\ref{sec:query}) algorithms. A summary of their experimental 
results can be found in Tables \ref{tab:preprocessing} and~\ref{tab:queries}.

\section{Preliminaries}
\label{sec:preliminaries}

A {\em directed graph} $G$ is a pair of a finite set $V(G)$ of vertices and 
a finite multiset $E(G) \subseteq V(G) \times V(G)$ of edges. The \emph{reverse
graph} $G^R$ of $G$ is a graph on the same set of vertices with all of the 
edges reversed. A \emph{subgraph} $H$ of a graph $G$ is denoted by $H \subseteq
G$. A subgraph $H \subseteq G$ is \emph{induced by a set of edges} $F \subseteq
E(G)$ if $E(H) = F$ and $V(H) = \{ u \in V(G) \, | \, \exists f\in F
\mbox{ incident with } u\}$; we then write~$H=G[F]$. 

A \emph{walk} $P \subseteq G$ is an alternating sequence $(u_0,e_1,u_1,\ldots,
e_k,u_k)$ of vertices and edges of $G$ such that $e_i = (u_{i-1},u_i)$ for 
$i = 1,\ldots,k$. A {\em subwalk} is a subsequence of a walk. 
A \emph{concatenation} $P_1 .\,P_2$ of walks $P_1=(u_0,e_1,
u_1,\ldots,e_k,u_k)$ and $P_2=(u_k,e_{k+1},u_{k+1},\ldots,e_l,u_l)$ is the walk 
$(u_0,e_1,u_1,\ldots,e_k,u_k,e_{k+1},\ldots,e_l,u_l)$. If $P_2$ is a single edge 
$f$, then we write~$P_1.\,f$. 

A walk $Q$ is is a {\em prefix} of another walk $P$ if $Q$ is a subwalk of
$P$ starting with the same index, and analogically with {\em suffix}.
The \emph{prefix set} of a walk $P=(u_0,e_1,\ldots,e_k,u_k)$ is $\mathit{Prefix}
(P)\stackrel{\textrm{\tiny{def}}}{=} \{(u_0,e_1,\ldots,e_i,u_i) |\> 0 \le i \le 
k\},$ and analogically $\mathit{Suf\!fix}(P) \stackrel{\textrm{\tiny{def}}}{=} 
\{(u_i,e_{i+1},\ldots,e_k,u_k) |\>  0 \le i\le k\}$. Two walks are {\em 
overhanging} (one another) if either one is a subwalk of the other, or a 
non-zero-length suffix of one is a prefix of the other (informally, one can 
traverse both with one superwalk).  

The \emph{weight} of a walk $P \subseteq G$ wrt. a weighting $w: E(G) 
\mapsto \mathbb{R}$ of $G$ is defined as $|P|_w=w(e_1)+w(e_2)+\dots+w(e_k)$
where $P=(u_0,e_1, \ldots,e_k,u_k)$. The \emph{distance} $\delta_w(u,v)$ from 
$u$ to $v$ in $G$ is the minimum weight over all walks from $u$ to $v$, or 
$\infty$ if there is no such walk.

A road network can be naturally represented by a graph $G$ such that the 
junctions are represented by $V(G)$ and the roads (or road segments) by 
$E(G)$. Chosen cost function (e.g. travel time, distance, expenses, etc.) 
is represented by a {\em non-negative} edge weighting $w: E(G) \mapsto
\mathbb{R}^+_0$ of $G$. 

\begin{definition}[Road Network]
\label{def:road_network}
Let $G$ be a graph with a non-negative edge weighting $w$. A \emph{road 
network} is the pair $(G,w)$.
\end{definition}

A brief overview of classical Dijkstra's and A* algorithms and their 
bidirectional variants for shortest paths follows, but we also would like
to recall the useful notion of a reach given by Gutman \cite{Gutman2004}. 

\begin{definition}[Reach \cite{Gutman2004}]
\label{def:reach}
Consider a walk $P$ in a road network $(G,w)$ from $s$ to~$t$ where 
$s,t\in V(G)$. The {\em reach} of a vertex $v\in V(P)$ on $P$ is $r_P(v)=\min 
\{ |P^{sv}|_w,|P^{vt}|_w \}$ where  $P^{sv}$ and $P^{vt}$ is a subwalk of $P$ from 
$s$ to $v$ and from $v$ to $t$, respectively. The \emph{reach of $v$ in $G$}, 
$r(v)$, is the maximum value of $r_Q(v)$ over all optimal walks $Q$ between 
pairs of vertices in $G$ such that $v\in V(Q)$.
\end{definition}

\subsubsection{Classical Shortest Paths Algorithms.} 
\label{sec:shortest-paths-algorithms}

Classical Dijkstra's algorithm solves the single source shortest paths
problem\footnote{Given a graph and a start vertex find the shortest paths 
from it to the other vertices.} in a graph $G$ with a non-negative 
weighting $w$. Let $s\in V(G)$ be the start vertex (and, optionally, let 
$t \in V(G)$ be the target vertex).

\begin{itemize}
\item The algorithm maintains, for all $v \in V(G)$, a {\em (temporary) 
  distance estimate} of the shortest path from $s$ to $v$ found so far in 
  $d[v]$, and a predecessor of $v$ on that path in $\pi[v]$. 

\item The scanned vertices, i.e. those with $d[v] = \delta_w(s,v)$ confirmed, 
  are stored in the set $T$; and the discovered but not yet scanned vertices, 
  i.e. those with $\infty >d[v] \geq \delta_w(s,v)$, are stored in the set $Q$. 

\item The algorithm work as follows: it iteratively picks a vertex $u \in Q$ 
  with minimum value $d[u]$ and relaxes all the edges $(u,v)$ leaving $u$.
  Then $u$ is removed from $Q$ and added to $T$. {\em Relaxing} an edge $(u,v)$
  means to check if a shortest path estimate from $s$ to $v$ may be improved 
  via $u$; if so, then $d[v]$ and $\pi[v]$ are updated. Finally, $v$ is added 
  into $Q$ if is not there already. 
\item The algorithm terminates when $Q$ is empty (or if $t$ is scanned).
\end{itemize}

Time complexity depends on the implementation of $Q$; such as it is 
${\cal{O}}(|E(G)| + |V(G)|\log|V(G)|)$ with the Fibonacci heap.

\smallskip

A* algorithm is also known as \emph{goal directed search} and it is equivalent 
to aforementioned Dijkstra's one on a graph modified as follows:

\begin{itemize}
\item A \emph{potential function} $p: V(G) \mapsto \mathbb{R}$ is defined and 
  used to reduce edge weights such that $w_p(u,v) = w(u,v) - p(u) + p(v)$.
  If $w_p$ is non-negative, the algorithm is correct and $p$ is called 
  \emph{feasible}.
\item Each path from $s$ to $t$ found by Dijkstra's algorithm in $G$ with 
  the reduced weighting $w_p$ then differs from reality by $p(t)-p(s)$ and this
  value must be subtracted at the end.
\end{itemize}

Dijkstra's and A* algorithms can be used ``bidirectionally'' to solve SPSP 
problem. Informally, one (forward) algorithm is executed from the start 
vertex in the original graph and another (reverse) algorithm is executed from 
the target in the reversed graph. Forward and reverse algorithms can 
alternate in any way and algorithm terminates, for instance, when there is 
a vertex scanned in both directions.

Dijkstra's and $A^*$ algorithms can be accelerated by reach as follows: 
when discovering a vertex $v$ from $u$, the algorithm first tests whether 
$r(v) \geq d[u] + w(u,v)$ (the current distance estimate from the start)
or $r(v) \geq {lower}(v)$ (an~auxiliary lower bound on the distance from $v$
to target), and only in case of success it inserts $v$ into the queue of 
vertices for processing.

\section{The New Concept -- Scope}
\label{sec:scope}

The main purpose of this section is to provide a theoretical foundation for 
the aforementioned vague objective of ``comfort of a route''. Recall that 
the scope levels referred in Definition~\ref{def:scope} are generally 
assigned according to auxiliary metadata of the road network, e.g.\ the 
road categories and additional available information which is presumably 
included with it; see Table~\ref{tab:scope_example}. Such a scope level 
assignment procedure is {\em not} the subject of the theoretical foundation.

\begin{definition}[Scope]
\label{def:scope}
Let $(G,w)$ be a road network. \emph{A scope mapping} is defined as ${\cal{S}}: 
E(G) \mapsto \mathbb{N}_0 \cup \{\infty\}$ such that $0,\infty \in 
Im({\cal{S}})$. Elements of the image $Im({\cal{S}})$ are called \emph{scope 
levels}. Each scope level $i\in Im({\cal{S}})$ is assigned a constant value of 
\emph{scope} $\nu^{\cal{S}}_i \in \mathbb{R}_0 \cup \{\infty\}$ such that 
$0 = \nu^{\cal{S}}_0 < \nu^{\cal{S}}_1 < \cdots < \nu^{\cal{S}}_\infty = \infty$.
\end{definition}

\begin{table}[t]
  \caption{A very simple demonstration of a scope mapping which is based 
    just on the road categories. Highways and other important roads have 
    unbounded ($\infty$) scope, while local, private or restricted roads have 
    smaller scope. The zero scope is reserved for roads that physically cannot 
    be driven through (including, for instance, long-term road obstructions). 
    The weight function of this example is the travel distance in meters.}
  \label{tab:scope_example}
  \centering
  \renewcommand{\arraystretch}{1.1}
  \begin{tabular*}{\textwidth}{c@{\mbox{\quad}}c@{\mbox{\quad}}c@{\
        \mbox{\quad}}c@{\mbox{\quad}}c}
    \hline \hline 
    \textit{Scope}& \textit{level} $i$ & \textit{Value $\nu^{\cal{S}}_i$} & 
    \textit{Handicap $\kappa^{\cal{S}}_i$} & \textit{Road category} \\
    \hline 
  &  0 & 0  & 1 & Alley, Walkway, Bike Path, Bridle Path \\
  &  1 & 250 & 50 & Parking Lot Road, Restricted Road \\
  &  2 & 2000 & 250 & Local Neighborhood Road, Urban Roads \\
  &  3 & 5000 & 600 & Rural Area Roads, Side Roads \\
  &  $\infty$ & $\infty$ & ($\infty$) & Highway, Primary (Secondary) Road \\
    \hline
  \end{tabular*}
\end{table}

To give readers a better feeling of how the scope level assignment outlined 
in Table~\ref{tab:scope_example} looks like, we present some statistics of 
the numbers of edges assigned to each level in 
Table~\ref{tab:scope_distribution}.

\begin{table}[hb]
  \caption{An example of the scope levels distribution in several road 
    networks graphs of the US; the initial scope is assigned in accordance with 
    road categories as described in Table~\ref{tab:scope_example} and then 
    balanced algorithmically.}
  \label{tab:scope_distribution}
  \renewcommand\arraystretch{1.1}
  \centering
  \begin{tabular*}{\textwidth}{@{\mbox{\qquad}}c@{\mbox{ }}@{\mbox{ }}r@{\mbox{\
        }}@{\mbox{ }}r@{\mbox{ }}@{\mbox{ }}r@{\mbox{ }}@{\mbox{ }}r@{\mbox{\
        }}@{\mbox{ }}r@{\mbox{ }}@{\mbox{ }}r@{\mbox{ }}}
    \hline \hline
    \textit{Scope level} & \textit{USA-all} & \textit{Alabama} & 
    \textit{Connecticut} & \textit{Florida} & \textit{Georgia} & 
    \textit{Indiana} \\
    \hline 
    $\infty$ & 8.171\% & 9.726\% & 12.874\% & 8.392\% & 11.331\% & 8.672\% \\
    $3$ & 62.420\% & 73.711\% & 64.022\% & 58.389\% & 69.564\% & 64.967\% \\
    $2$ & 19.712\% & 10.052\% & 17.018\% & 28.687\% & 13.250\% & 17.270\% \\
    $1$ & 6.570\% & 5.933\% & 3.978\% & 2.038\% & 4.946\% & 7.239\% \\
    $0$ & 0.660\% & 0.018\% & 0.092\% & 0.083\% & 0.020\% & 0.007\%\\
    \hline
  \end{tabular*}
\end{table}

There is one more formal ingredient missing to make the scope concept a 
perfect fit: imagine that one prefers to drive a major highway, then she 
should better not miss available slip-roads to it. This is expressed with 
a~``handicap'' assigned to the situations in which a turn to a next road of 
higher scope level is possible, as follows:

\begin{definition}[Turn-Scope Handicap]
  \label{def:turnhandicap}
  Let ${\cal S}$ be a scope mapping in $(G,w)$. The \emph{turn-scope handicap} 
  $h_{\cal S}(e)\in\mathbb{N}_0 \cup \{\infty\}$ is defined, for every 
  $e=(u,v)\in E(G)$, as the maximum among ${\cal S}(e)$ and all ${\cal S}(f)$ 
  over $f=(v,w)\in E(G)$. Each handicap level $i$ is assigned a constant 
  $\kappa^{\cal{S}}_i$ such that $0<\kappa^{\cal{S}}_0 < \cdots <\kappa^{\cal{S}}_\infty$.
\end{definition}

The desired effect of admitting low-level roads only ``near'' the start
or target positions---until higher level roads become widely available---is
formalized in next Def.~\ref{def:xadmissible},~\ref{def:stadmissible}.
We remark beforehand that the seemingly complicated formulation is actually
the right simple one for a smooth integration into Dijkstra.

\begin{definition}[Scope Admissibility]
  \label{def:xadmissible}
  Let $(G,w)$ be a road network and let $x \in V(G)$. An edge $e = (u,v) \in 
  E(G)$ is \emph{$x$-admissible} in $G$ for a scope mapping $\cal{S}$ if, and 
  only if, there exists a walk $P \subseteq G - e$ from $x$ to $u$ such that
  \vspace{-1ex}
  \begin{enumerate}
    \item each edge of $P$ is $x$-admissible in $G - e$ for $\cal{S}$,
    \item $P$ is optimal subject to (1), and
    \item for $\ell={\cal{S}}(e)$,~
	$\sum_{f \in E(P),\, {\cal{S}}(f)>\ell}\, w(f) \,+\, \sum_{f\in 
        E(P),\, h_{\cal S}(f) > \ell \geq {\cal{S}}(f)}\,
	\kappa^{\cal{S}}_{\ell} \,\le\> \nu^{\cal{S}}_{\ell}$.
  \end{enumerate}
\end{definition}

Note; every edge $e$ such that ${\cal S}(e)=\infty$ ({\em unbounded scope} 
level) is always admissible, and with the values of $\nu^{\cal{S}}_i$ growing 
to infinity, Def.~\ref{def:xadmissible} tends to admit more and more edges 
(of smaller scope).

\begin{definition}[Admissible Walks]
  \label{def:stadmissible}
  Let $(G,w)$ be a road network and $\cal{S}$ a scope mapping. For a walk
  $P=(s=u_0,e_1,\dots e_k,u_k=t)\subseteq G$ from $s$ to~$t$;
  \vspace{-1ex}
  \begin{itemize}
  \item $P$ is \emph{$s$-admissible} in $G$ for $\cal{S}$ if every $e_i 
    \in E(P)$ is $s$-admissible in $G$ for $\cal{S}$,
  \item and $P$ is \emph{$st$-admissible} in $G$ for $\cal{S}$ if there 
    exists $0\leq j\leq k$ such that every $e_i \in E(P)$, $i\leq j$ is 
    $s$-admissible in $G$, and the reverse of every $e_i\in E(P)$, $i>j$ 
    is $t$-admissible in $G^R$ -- the graph obtained by reversing all edges.
  \end{itemize}
\end{definition}

\subsubsection{Proper Scope Mapping.}

In a standard connectivity setting, a graph (road network) $G$ is 
{\em routing-connected\/} if, for every pair of edges $e,f\in E(G)$, there 
exists a walk in $G$ starting with $e$ and ending with $f$. This obviously 
important property can naturally be extended to our scope concept as follows.

\begin{definition}[Proper Scope]
  \label{def:properscope}
  A scope mapping $\cal{S}$ of a routing-connected graph $G$ is {\em proper} 
  if, for all $i\in Im({\cal{S}})$, the subgraph $G^{[i]}$ induced by those 
  edges $e\in E(G)$ such that ${\cal S}(e)\geq i$ is routing-connected, too. 
\end{definition}

Note that validity of Definition~\ref{def:properscope} should be enforced in 
the scope-assignment phase of preprocessing (e.g., the assignment should 
reflect known long-term detours on a highway\footnote{Note, regarding a
real-world navigation with unexpected road closures, that the pro\-per-scope 
issue is not at all a problem---a detour route could be computed from 
the spot with ``refreshed'' scope admissibility constrains. Here we solve the 
static case.} accordingly). The topic of connectivity requirements for 
assigned scope levels deserves a bit closer explanation. A good scope 
mapping primarily reflects the road quality (given by its category and 
other attributes) as given in the road network metadata. This simple 
assignment is, however, not directly usable in practice due to rather 
common errors in road network data. Such errors typically display themselves 
as violations of connectivity at different scope levels, i.e.\ as 
violations to Def.~\ref{def:properscope}. Hence these errors can be 
algorithmically detected and subsequently repaired (preferably with help 
of other road network metadata) so that the final scope mapping conforms 
to Def.~\ref{def:properscope}--being a proper scope. 

\vspace{-2ex}
\subsubsection{Existence of admissible walks. } A natural conclusion of 
routing connectivity (proper scope) of a road network then reads 

\begin{theorem}
  Let $(G,w)$ be a routing-connected road network and let $\cal{S}$ 
  be a~proper scope mapping of it. Then, for every two edges $e=(s,x)\in E(G)$,
  $f=(y,t)\in E(G)$, there exists a $st$-admissible walk $P\subseteq G$
  for $\cal S$ such that $P$ starts with the edge $e$ and ends with~$f$. 
\end{theorem}

\begin{proof}
We refer to the notation of Definition~\ref{def:properscope}. Let 
$i=\min({\cal S}(e),{\cal S}(f))$, and so $e,f\in G^{[i]}$. Since $\cal S$ is 
proper, there exists a walk $P^i\subseteq G^{[i]}$ starting with the edge $e$ 
and ending with $f$, and $|P^i|_w<\infty$. If $i=\infty$, then we are done
since every edge of $G^{[\infty]}$ is automatically $st$-admissible (regardless 
of $s,t$). So, by means of contradiction, we assume that $i\in\mathbb N\cap 
Im({\cal S})$ is highest possible such that the theorem failed, i.e.\ no such 
$P^i$ is $st$-admissible.

\medskip 

Let $P^i=(v_0,e_1,\dots,e_k,v_k)$. We find maximal index $1\leq p< k$ such that
${\cal S}(e_p)\geq i+1$ or that $h_{\cal S}(e_p)\geq i+1$ (where the value of 
$h_{\cal S}(e_p)$ is witnessed by an edge $f_0$ of ${\cal S}(f_0)\geq i+1$ starting
in $v_p$). Informally, $p$ is the least position on $P^i$ where the scope 
admissibility condition gets affected at level~$i$\,---see the formula in 
Def.~\ref{def:xadmissible}. Since $P^i$ is not $st$-admissible by our 
assumption, this $p$ is well defined. In the first case, ${\cal S}(e_p)\geq 
i+1$, we set $e'=e_p$. In the second case, ${\cal S}(e_p)\leq i$ and 
${\cal S}(f_0)\geq i+1$, we set $e'=f_0$. We also denote by $s'$ the tail of 
$e'$. We symmetrically define $f'$ and $t'$ in the reverse walk of $P^i$ in
$G^R$.

\medskip 

Now $j=\min({\cal S}(e'),{\cal S}(f'))\geq i+1$. Since $\cal S$ is proper, 
there exists a walk $P^j\subseteq G^{[j]}$ starting with the edge $e'$ and 
ending with $f'$, and $|P^j|_w<\infty$. Moreover, by maximality of our choice 
of $i$, this $P^j$ can be chosen such that $P^j$ is $s't'$-admissible. Then 
we define $P'\subseteq P^i\cup P^j$ as the walk which starts and ends as 
$P^i$ while using $P^j$ in the section ``between'' $e'$ and $f'$. By the 
least choice of $e'$ ($f'$) above, all edges of $P'$ are clearly also 
$st$-admissible (which follows from admissibility of all edges of $P^j$ at 
levels $>i$ there). The proof is finished.
\qed
\end{proof}

\subsubsection{Well distributed scope. } There is another, more vague, 
technical requirement on a useful scope mapping, which becomes particularly 
important in the opening cellular stage of the query algorithm in 
Section~\ref{sec:query}: for each bounded scope level there should be no 
relatively large areas (subnetworks) not penetrated by any road of higher 
scope level. This requirement is formally described as that the subgraph
$G-V(G^{[i]})$, i.e.\ the subgraph induced by those vertices incident only 
with edges of scope level $<i$, has no relatively large connected components 
for each $i\in Im({\cal S})$. Such as, for $i=\infty$ the words ``no relatively
large'' mean no component of size significantly larger than the expected cell 
size. For smaller values of $i$ the size limit on components is accordingly 
smaller. We then say that the scope mapping $\cal S$ is {\em well-distributed} 
in $G$.

Again, as in the case of proper scope, the requirement of having 
well-distributed scope mapping is not strictly a subject of the theoretical 
foundation of scope. Instead, wisely designed road networks (with the 
corresponding metadata) should conform to such a requirement automatically.
In other words, if a scope level assignment does not produce a well-distributed
scope mapping, then there is something wrong; either directly in the network 
design, or in the provided metadata (or in the way we understand it). Yet, 
violations of the well-distributivity property can be easily detected in an 
algorithmic way, and also automatically repaired via taking some heuristically 
selected shortest walks across the large component, and ``upgrading'' them 
into higher scope levels.

\subsection{ $\cal{S}$-Dijkstra's Algorithm and $\cal{S}$-Reach}
\label{sub:Sdijkstra}

As noted beforehand, Def.~\ref{def:stadmissible} smoothly integrates into
bidirectional Dijkstra's or A* algorithm, simply by keeping track of the
admissibility condition (3.):

\begin{description}
\item[$\cal{S}$-Dijkstra's Algorithm {\rm(one direction of the search).}]
\item[~\,--]\smallskip
  For every accessed vertex $v$ and each scope level $\ell\in Im({\cal{S}})$,
  the algorithm keeps, as $\sigma_\ell[v]$, the best achieved value 
  of the sum.\ formula in Def.~\ref{def:xadmissible}\,(3.).
\item[~\,--] The $s$-admissibility of edges $e$ starting in $v$ depends then
  simply on $\sigma_{{\cal S}(e)}[v]\leq \nu^{\cal{S}}_{{\cal S}(e)}$, and only 
  $s$-admissible edges are relaxed further.
\end{description}

The full details follow here, in Algorithm~\ref{alg:SDijkstra}. Notice 
that, although a route planning position in a road network is generally 
an edge (not just a vertex) there is no loss of generality if we consider 
only vertices as the start and target positions of a route planning 
query. In the general case when a position is an edge $f$, we may simply 
subdivide $f$

\begin{algorithm}[H]
\caption{~$\cal{S}$-Reach-based $\cal{S}$-Dijkstra's Algorithm}
\label{alg:SDijkstra}
\begin{algorithmic}[1]    
\smallskip
\REQUIRE A road network $(G,w)$, a scope ${\cal{S}}$, a start vertex 
  $s \in V(G)$, an $\cal{S}$-reach $r^{\cal{S}}$.
\smallskip
\ENSURE For every $v \in V(G)$, an optimal $s$-admissible walk from $s$ to $v$ 
  in $G$ (or $\infty$).
\end{algorithmic}
\smallskip
\underline{\textsc{Relax}$(u,v,\gamma)$} \vskip 2pt
\begin{algorithmic}[1]    
  \IF[\hfill// Temporary distance estimate updated.]{$d[u] + w(u,v) < d[v]$}
    \STATE $Q \is Q \cup \{v\}$
    \STATE $d[v] \is d[u] + w(u,v)$;~$\pi[v] \is u$
  \ENDIF
  \IF[\hfill// Scope admissibility vector updated.]{$d[u] + w(u,v) \le d[v]$}
    \FORALL{$i \in Im({\cal{S}})$}
      \STATE $\sigma_i[v] \leftarrow \min\{\sigma_i[v],\>\sigma_i[u]+\gamma_i\}$ 
    \ENDFOR
  \ENDIF
  \RETURN
\end{algorithmic}

\medskip
\underline{\textsc{${\cal{S}}$-Dijkstra}$(G,w,{\cal{S}},s,r^{\cal{S}})$} 
\vskip 3pt
\begin{algorithmic}[1]
 \FORALL[\hfill// Initialization.]{$v \in V(G)$}
  \STATE $d[v] \is \infty$;~$\pi[v] \is \bot$;~
  \COMMENT{\hfill// Distance estimate and predecessor.}
  \STATE $\sigma[v] \leftarrow (\infty,\ldots,\infty)$
  \COMMENT{\hfill// Scope admissibility vector.}
 \ENDFOR
 \STATE $d[s] \is 0$;~ $Q \is \{s\}$;~$\sigma[s] \is (0,\ldots,0)$
 \WHILE[\hfill// Main loop processing all vertices.]{$Q \neq \emptyset$}
  \STATE $u \is  \min_{d[]}(Q)$;~ $Q \is Q \setminus\{u\}$
  \COMMENT{\hfill// Pick a vertex $u$ with the minimum $d[u]$.}
  \FORALL[\hfill// All edges from $u$; subject to]{$f=(u,v)\in E(G)$}
   \IF[\hfill// $\cal{S}$-reach check, and]{$r^{\cal{S}}(v) \geq d[u] + w(f)
	 \>\lor \> {\cal{S}}(f) = \infty$}
     \IF[\hfill// $s$-admissibility check.]{$\sigma_{{\cal S}(f)}[u]\leq 
       \nu^{\cal{S}}_{{\cal S}(f)}$}
      \FORALL[\hfill// Adjustment to scope admissibility.]{$i\in Im({\cal{S}})$}
        \OLELIF{${\cal{S}}(f)>i$}{$\gamma_i\is w(f)$}{$\gamma_i \is 0$}
        \OLIF{$h_{\cal S}(f)>i\geq{\cal S}(f)$}
             {$\gamma_i\is\gamma_i+\kappa^{\cal{S}}_{i}$}
      \ENDFOR
      \STATE \textsc{Relax}$(u,v,\gamma)$
      \COMMENT{\hfill// Relaxation of $f=(u,v)$.}
     \ENDIF
    \ENDIF
  \ENDFOR
 \ENDWHILE
 \STATE \textsc{ConstructWalk}$\,(G,d,\pi)$
 \COMMENT{\hfill// Postprocessing -- generating output.}
\end{algorithmic}
\end{algorithm}

\begin{theorem}
\label{thm:SDijkstra}
{\em$\cal{S}$-Dijkstra's algorithm}, for a road network $(G,w)$, a scope 
mapping ${\cal{S}}$, and a start vertex $s \in V(G)$, computes an optimal 
$s$-admissible walk from $s$ to every $v \in V(G)$ in time ${\cal O}\big(
|E(G)|\cdot|Im({\cal{S}})|+|V(G)|\cdot\log |V(G)|\big)$.
\end{theorem}

\begin{proof}
We divide the proof of correctness of the algorithm into three steps:

\begin{enumerate}\parskip2pt
\item Assume that---in line 10---the algorithm correctly identifies the
  $s$-admissible edges $f$, and that---in line 9---no edge $f$ is discarded
  which is a part of some optimal $s$-admissible walk between the start $s$ and 
  any other $s$-saturated vertex $y$ (in $G$ for~$\cal S$). Then we can directly
  use a traditional proof of ordinary Dijkstra's algorithm to argue that the 
  computed results are optimal $s$-admissible walks starting from $s$ in the 
  network $(G,w)$ for scope $\cal S$.
  
  As to the claimed complexity bound, now it is enough to add a factor of 
  $|Im({\cal{S}})|$ (which can be regarded as a constant) for the loop on 
  line 11.

\item Suppose the assumption on line~9 got potentially violated, i.e.\ the
  particular edge $f=(u,v)$ of bounded scope belongs to an optimal 
  $s$-admissible walk $Q$ from $s$ to some $s$-sat\-urated $y$. Let 
  $Q'\subseteq Q$ be a subwalk starting with $x=s$ and ending with the edge 
  $f$ (i.e., in the vertex $v$). By Definition~\ref{def:Sreach}, the $\cal 
  S$-reach of $v$ is defined and at least $|Q'|_w\geq d[u]+w(f)$, and so the 
  condition on line~9 is evaluated as TRUE. The claim is done.

\item Hence it is enough to prove that the scope-admissibility vector 
  $\sigma$ is computed such that an edge $f=(u,v)$ starting in a vertex $u$ is 
  $s$-admissible if and only if 
  $\sigma_{{\cal S}(f)}[u]\leq \nu^{\cal{S}}_{{\cal S}(f)}$.
  By Definition~\ref{def:xadmissible}, and induction on the length (number of
  edges) of the discovering walk from $s$ to $u$; it is just enough to 
  straightforwardly verify that the adjustments to $\sigma$ accumulated in 
  lines 12--13 exactly correspond to the summation formula in 
  Definition~\ref{def:xadmissible}\,(3), and employ the known fact that $d[u]$ 
  is an optimal weight of a walk to~$u$. \qed
\end{enumerate}
\end{proof}

Furthermore, practical complexity of this algorithm can be largely decreased 
by a suitable adaptation of the {\em reach} concept (Def.~\ref{def:reach}),
given in Def.~\ref{def:Sreach}. For $x \in V(G)$ in a~road network with 
scope $\cal{S}$, we say that a vertex $u \in V(G)$ is $x$-\emph{saturated} 
if no edge $f=(u,v)$ of $G$ from $u$ of {bounded scope} (i.e., ${\cal S}
(f)<\infty$) is $x$-admissible for $\cal S$. A~walk $P$ with ends $s,t$ 
is {\em saturated} for $\cal S$ if some vertex of $P$ is both $s$-saturated 
in $G$ and $t$-saturated in the reverse network $G^R$.

\begin{definition}[$\cal{S}$-reach]
\label{def:Sreach}
Let $(G,w)$ be a road network and $\cal{S}$ its scope mapping. The 
\emph{$\cal{S}$-reach of $v \in V(G)$ in $G$}, den.~$r^{\cal{S}}(v)$, is the 
maximum value among $|P^{xv}|_w$ over all $x,y \in V(G)$ such that $y$ is 
$x$-saturated while $v$ is not $x$-saturated, and there exists an optimal 
$x$-admissible walk $P\subseteq G$ from $x$ to $y$ such that $P^{xv}$ is a 
subwalk of $P$ from $x$ to $v$. ~$r^{\cal{S}}(v)$ is undefined ($\infty$) if 
there is no such walk.
\end{definition}

There is no general easy relation between classical reach and $\cal{S}$-reach; 
they both just share the same conceptual idea. Moreover, $\cal{S}$-reach 
can be computed more efficiently (unlike reach) since the set of 
non-$x$-saturated vertices is rather small and local in practice, and only 
its $x$-saturated neighbors are to be considered among the values of $y$ 
in Def.~\ref{def:Sreach}.

The way $\cal{S}$-reach of Def.~\ref{def:Sreach} is used to amend 
$\cal{S}$-Dijkstra's algorithm is again rather intuitive; an edge 
$f=(u,v)\in E(G)$ is relaxed from $u$ only if ${\cal S}(f)=\infty$, or 
$r^{\cal{S}}(v)\geq d[u]+w(f)$. 

\begin{theorem}
\label{thm:SDijkstraA}
Assume {\em$\cal{S}$-Reach-based $\cal{S}$-Dijkstra's algorithm} (Alg. 
\ref{alg:SDijkstra}), with a road network $(G,w)$, a scope mapping ${\cal{S}}$,
and any upper bound $r^{\cal{S}}$ on the $\cal{S}$-reach in $G$. This algorithm, 
for a start vertex $s \in V(G)$ as its input, computes, for every $v \in V(G)$,
an optimal $s$-admissible walk from $s$ to $v$ in $G$ in time 
${\cal O}\big(|E(G)|\cdot|Im({\cal{S}})|+|V(G)| \cdot\log|V(G)|\big)$.
\end{theorem}

\begin{remark}
Notice that we have used an ``upper bound on the $\cal{S}$-reach'' in the
statement, although we can easily compute the $\cal{S}$-reach exactly.
The deeper reason for doing this generalization is that in a practical case of
multiple utility weight functions for the same road network, we may want to
store just one maximal instance among the $\cal{S}$-reach values for every
vertex in the network (instead of a separate value per each utility
function).
\end{remark}

However, our $\cal{S}$-reach amending scheme has one inevitable 
limit of usability---it becomes valid only if the both directions of Dijkstra's 
search get to the ``saturated'' state. (In the opposite case, the start and 
target are close to each other in a local neighborhood, and the shortest 
route is quickly found without use of $\cal{S}$-reach, anyway.) Hence we 
conclude:

\begin{theorem}
\label{thm:Sreach}
Let $s,t\in V(G)$ be vertices in a road network $(G,w)$ with a scope mapping
${\cal{S}}$. Bidirectional {\em$\cal{S}$-reach $\cal{S}$-Dijkstra's algorithm}
computes an optimal one among all $st$-admissible walks in $G$ from $s$ to~$t$
which are saturated for~$\cal S$.
\end{theorem}

Notice that this strengthening Theorem~\ref{thm:Sreach} is directly implied 
by Theorem~\ref{thm:SDijkstraA} which also directly implies former Theorem 
\ref{thm:SDijkstra} (when setting $r^{\cal S}\equiv\infty$).

\section{The Route-Planning Algorithm -- Preprocessing}
\label{sec:algorithm}

Following the informal outline from the introduction, we now present the 
second major ingredient for our approach; a separator based partitioning of 
the road network graph with respect to a given scope mapping.

\subsection{Partitioning into Cells}
\label{sub:preprocessing}

At first, a road network is partitioned into a set of pairwise edge-disjoint
subgraphs called \emph{cells} such that their \emph{boundaries} (i.e., the
vertex-separators shared between a cell and the rest) contain as few as 
possible vertices incident with edges of unbounded scope. The associated 
formal definition follows.

\begin{definition}[Partitioning and Cells]
  \label{def_cell}
  Let ${\cal{E}} = \{E_1,\ldots,E_\ell\}$ be a partition of the edge set $E(G)$
  of a graph $G$. We call {\em cells} of $(G,{\cal E})$ the subgraphs $C_i = 
  G[E_i] \subseteq G$, for $i = 1,2, \ldots, \ell$. The {\em cell boundary} 
  $\Gamma(C_i)$ of $C_i$ is the set of all vertices that are incident both 
  with some edge in $E_i$ and some in $E(G)\setminus E_i$, and the 
  {\em boundary} of $\cal E$ is $\Gamma(G,{\cal E})=\bigcup_{1\leq i\leq\ell} 
  \Gamma(C_i)$.
\end{definition}

Practically, we use a graph partitioning algorithm hierarchically computing a 
so-called \emph{partitioned branch-decomposition} of the road network. The 
algorithm employs an approach based on max-flow min-cut which, though being 
heuristic, performs incredibly well---being fast in finding really good small 
vertex separators.\footnote{It is worth to mention that max-flow based 
heuristics for a branch-decomposition have been used also in other 
combinatorial areas recently, e.g.\ in the works of Hicks.} 

\subsubsection{Partitioning Algorithm.}
\label{app:partitioning}

Notice that we are decomposing the whole graph $G$ and not only its 
unbounded-scope subgraph $G^{[\infty]}$, this is because of the cellular stage 
of our query algorithm. A~brief outline of the partitioning method follows.

\begin{enumerate}
\item \emph{Simplification step.} At first, all weakly disconnected components 
  are placed into a special ``disconnected'' cell and then removed from the 
  road network to ensure that the rest is weakly-connected. Then the road 
  network is contracted as much as possible so that all self-loops and parallel
  edges are removed, each maximal induced subtree is cut away, all maximal 
  induced subpaths are contracted into single edges. This way we reduce the 
  size of the road network by 10\%, approximately.

  \medskip

\item \emph{Replacement step.} We are interested in vertex cuts and thus,
  prior to max-flow min-cut computation, we have to replace vertices by new 
  edges as follows: each vertex $v$ in the road network is replaced by a 
  new edge $(v_i,e_v,v_o)$ such that every edge with its end in $v$ has now 
  its end in $v_i$ and every edge with its start in $v$ has now its start in 
  $v_o$. Notice that any cut consisting of such new edges represents a vertex 
  cut in the original road network.

  \medskip

\item \emph{Capacity assignment.} In general, each newly introduced edge has
  unit capacity and capacities of the other edges are set to positive infinity 
  so that they will never be saturated during the max-flow computation and 
  hence only new edges could appear in a cut (so that they form a vertex cut
  in the original road network). It is also beneficial (to allow a certain 
  level of parallelism) to cut the graph into two subgraphs of~similar sizes 
  and hence we choose sufficiently distant vertices to be the source and the 
  sink and, moreover, every edge within some ``restricted distance'' from the 
  source or sink gets infinite capacity to prevent its saturation (this way we 
  can ensure that a cut appears somewhere in the ``middle'' between the source 
  and sink). Also, to minimize the size (the number of edges) of the strict 
  boundary graph we add some handicap to the capacities of those new edges 
  which are adjacent to network edges of unbounded scope level.

  \medskip

\item \emph{Max-flow min-cut computation.} As soon as all capacities are 
  assigned, the maximum flow is computed and we obtain a minimum vertex cut in 
  the original road network. Hence we cut the original road network into two 
  subnetworks (of similar size), i.e.\ into two super-cells. The whole process
  now continues from step (1) concurrently in these subnetworks until all
  cells have suitable size.

\end{enumerate}

Furthermore, our algorithm always leverages natural disposition of a road 
network such that each sufficiently big autonomous region (i.e. the US state)
is partitioned separately. The resulted partitions are then ``glued'' 
together using natural borders. Clearly, this process might create too big 
boundaries and then the algorithm unions adjacent borderline cells and tries 
to partition them better. Such a separate partitioning makes it possible to 
employ a very beneficial parallelism. 

\subsection{Bounday Graph Construction}

Secondly, the in-cell distances between pairs of boundary vertices are 
precomputed such that only the edges of unbounded scope are used. This 
simplification is, on one hand, good enough for computing optimal 
routes on a ``global level'' (i.e., as saturated for scope in the sense of
Sec.~\ref{sub:Sdijkstra}).
On the other hand, such a simplified precomputed distance graph 
(cf.~$B_{\cal E}$)
is way much smaller than if all boundary-pair distance were stored for each 
cell. See in Table~\ref{tab:preprocessing}, the last column.

We again give the associated formal definition and a basic statement whose
proof is trivial from the definition.

\begin{definition}[Boundary Graph; $\cal S$-restricted]
 \label{def:boundary}
 Assume a road network $(G,w)$ together with a partition ${\cal E}=\{E_1,
 \ldots,E_\ell\}$ and the notation of Def.~\ref{def_cell}. For a scope mapping 
 $\cal S$ of $G$, let $G^{[\infty]}$ denote the subgraph of $G$ induced by the 
 edges of unbounded scope level $\infty$, and let $C_i^{\infty}=G[E_i]\cap 
 G^{[\infty]}$.

 The {\em ($\cal S$-restricted) in-cell distance graph} $D_i$ of the cell $C_i$
 is defined on the vertex set $V(D_i)=\Gamma(C_i)$ with edges and weighting 
 $p_i$ as follows. For $u,v\in \Gamma(C_i)\cap V(C_i^{\infty})$ only, let 
 $\delta^\infty_w(u,v)$ be the distance in $C_i^{\infty}$ from $u$ to $v$, and let 
 $f=(u,v) \in E(D_i)$ iff $p_i(f):=\delta^\infty_w(u,v)<\infty$.

 \smallskip
 The weighted {\em ($\cal S$-restricted) boundary graph $B_{\cal E}$} of a road 
 network $(G,w)$ wrt.\ scope mapping $\cal S$ and partition $\cal{E}$
 is then obtained as the union of all the cell-distance graphs $D_i$ for 
 $i = 1,2,\ldots,\ell$, simplified such that for each bunch of parallel edges
 only one of the smallest weight is kept in $B_{\cal E}$.
\end{definition}

\begin{proposition}
\label{prop:BEdistance}
  Let $(G,w)$ be a road network, $\cal S$ a scope mapping of it, and ${\cal{E}}$
  a partition of $E(G)$. For any $s,t\in \Gamma(G,{\cal E})$, the minimum 
  weight of a walk from $s$ to $t$ in $G^{[\infty]}$ equals the distance from $s$
  to $t$ in $B_{\cal E}$.
\end{proposition}

\subsection{Experimental Evaluation.} 

The prototype of the preprocessing algorithm is written in C and uses 
Ford-Fulkerson's max-flow min-cut algorithm and cells of approximately 
5000 edges. Minimum possible cell size is hard-coded to 2000 and maximum 
to 10000 edges. Publicly available road networks are taken from 
TIGER/Line 2010 \cite{Tiger2009} published by the U.S. Census Bureau using 
directed edges, i.e. every traffic lane is represented by one edge. The 
compilation is done by gcc 4.3.2 with -O2, and the preprocessing has been 
executed on a quad-core XEON machine, with 16GB RAM running Debian 5.0.4, 
GNU/Linux 2.6.26-2-xen-amd64. 

We have run our partitioning algorithm (Def.~\ref{def_cell}), in-cell distance 
computations (Def.~\ref{def:boundary}), and $\cal S$-reach computation 
(Def.~\ref{def:Sreach}) in parallel on a quad-core XEON machine with 16\,GB 
in 32 threads. A decomposition of the continental US road network into the 
boundary graph $B_{\cal E}$, together with computations of $\cal S$-reach in~$G$
(and $\cal S$-reach in $G^R$), distances in $D_i$s, and standard reach estimate
 for $B_{\cal E}$, took only \emph{192 minutes altogether}. This, and the tiny 
size of $B_{\cal E}$, are both very for potential practical applications in 
which the preprocessing may have to be run and the small boundary graphs separately stored for multiple utility weight functions (while the $\cal S$-reach 
values could still be kept in one maximizing). 

The collected data are briefly summarized in Table~\ref{tab:partitioning} and 
Table~\ref{tab:preprocessing}. To get a taste of the topic, some examples of 
local cell partitioning results are depicted in Fig. \ref{fig:preprocessing1} 
and~\ref{fig:preprocessing2}.


\begin{table}[H]
  \caption{Partitioning results for the TIGER/Line 2009 \cite{Tiger2009} US 
    road network. The left section identifies the (sub)network and its size, 
    the middle one the numbers and average size of partitioned cells, and the
    right section summarizes the results---the boundary graph size data,
    with percentage of the original network size. This boundary graph 
    $B_{\cal E}$ is wrt.\ the simple scope assignment of Table~\ref{def:scope}, 
    and $V(B_{\cal E})$ includes also isolated boundary vertices (i.e.\ those 
    with no incident edge of unbounded scope). Notice the tiny size ($<\!1$\%) 
    of $B_{\cal E}$ compared to the original road network, and the statistics
    regarding {\em cell boundary sizes}: maximum is 74, average 19, median
    18, and 9-decil is 31 vertices.}
  \label{tab:preprocessing}
  \centering
  \begin{tabular}{@{~}l@{\quad}r@{\quad}|@{\quad}r@{\quad}c@{\quad}|@{\quad}r@{\
        \quad}r}
    \hline \hline 
    \multicolumn{2}{c|@{\quad}}{\it Input $G$} & 
	\multicolumn{2}{c|@{\quad}}{\it Partitioning} & 
	\multicolumn{2}{c}{\it Boundary Graph $B_{\cal E}$} \\
    \footnotesize{Road network}    & \footnotesize{\#Edges~~}  & 
    \footnotesize{\#Cells} &  \footnotesize{Cell sz.}  & 
    \footnotesize{\#Vertices} & \footnotesize{\#Edges /~\%~size~} \\
    \hline
    USA-all    & 88 742 994 & 15 862  & 5 594 &  253 641 & 
    524 872 / 0.59\% \\
    USA-east   & 24 130 457 &  4 538  & 5 317 &  62 692 & 
    107 986 / 0.45\% \\
    USA-west   & 12 277 232 &  2 205  & 5 567 &  23 449 &  
    42 204 / 0.34\% \\
    Texas      &  7 270 602 &  1 346  & 5 366 &  17 632 &  
    36 086 / 0.50\% \\
    California &  5 503 968 &  1 011  & 5 444 &  11 408 &  
    16 978 / 0.31\% \\
    Florida    &  3 557 336 &    662  & 5 373 &  5 599 &  
    25 898 / 0.73\% \\
    \hline
  \end{tabular}
   \vspace{-1ex}
\end{table}

\begin{figure}[ht!]
  \centering
  \centerline{\epsfig{file=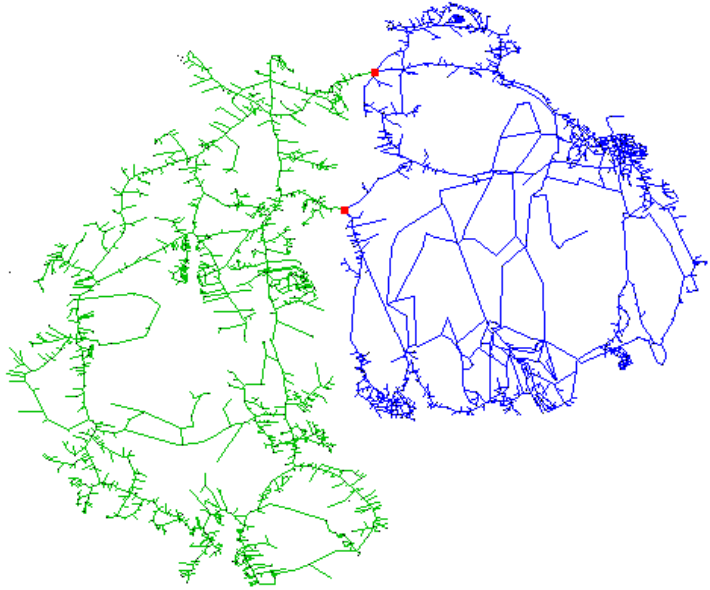, scale=.57}}
  \caption{The road network partitioned into two cells (blue and green) with
    boundaries of size two (red), notice that each walk starting in a cell and 
    ending in another must pass the boundary. Acadia National Park, Mount 
    Desert, ME, United States.}
  \label{fig:preprocessing2}
\end{figure}

\begin{figure}[H]
  \centering
  \centerline{\epsfig{file=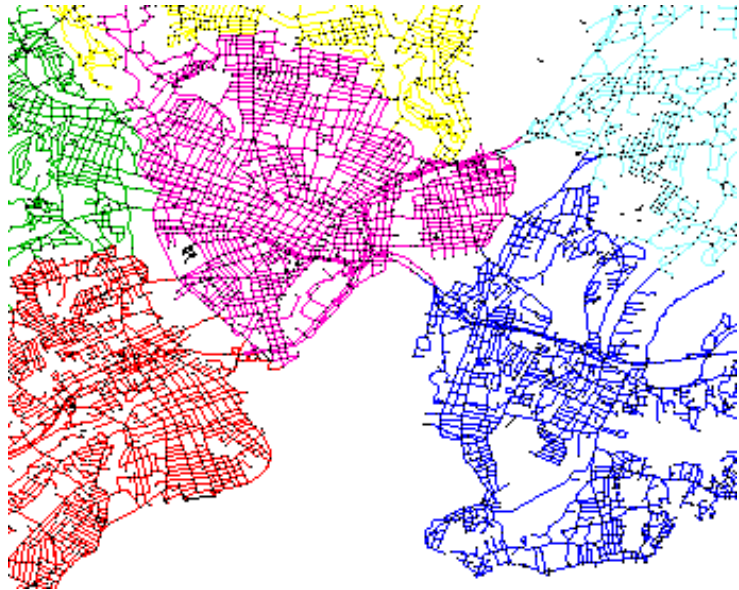, scale=0.8}}
  \caption{The partition computed using a max-flow algorithm respects the 
    natural road network disposition so that cells boundaries often lie along 
    the rivers, railways, canyons, etc. This, in particular, leads to smaller 
    boundaries. New Haven, CT, United States.}
  \label{fig:preprocessing1}
  \vspace{-3ex}
\end{figure}

\begin{table}[H]
  \caption{The example of the partitioning of several US states.}
  \label{tab:partitioning}
  \renewcommand\arraystretch{1.1}
  \centering
  \begin{tabular*}{\textwidth}{@{\quad}c@{\qquad}c@{\qquad}c@{\qquad}c@{\qquad}cc@{\qquad}}
    \hline \hline
    \textit{Road network} & \textit{\#Vertices} & \textit{\#Edges} & 
    \textit{\#Cells} & \textit{Avg. cell size}   & \textit{Time}  \\
    \hline
    Alabama   & 811 434 & 1 926 052 & 338 & 5 698 & 11 min \\
    Indiana   & 800 295 & 1 975 898 & 372 & 5 311 & 16 min \\
    Michigan  & 860 421 & 2 145 960 & 399 & 5 378 & 61 min \\
    Minnesota  & 908 292 & 2 166 138 & 383 & 5 666 & 18 min \\
    Arizona   & 893 163 & 2 184 866 & 420 & 5 202 & 25 min \\
    Georgia   & 945 212 & 2 226 392 & 400 & 5 655 & 14 min \\
    New York  & 890 684 & 2 236 530 & 422 & 5 299 & 47 min \\
    North Carolina & 1 104 258 & 2 497 764 & 409 & 6 107 & 12 min \\
    Oklahoma     & 1 049 680 & 2 508 862 & 436   & 5 754 & 23 min \\
    Ohio         & 1 085 287 & 2 761 920 & 483   & 5 404 & 38 min \\
    Illinois     & 1 085 287 & 2 761 920 & 530   & 5 211 & 21 min \\
    Missouri     & 1 291 751 & 3 020 152 & 519   & 5 891 & 34 min \\
    Pennsylvania & 1 263 737 & 3 081 096 & 574   & 5 367 & 26 min \\
    Virginia     & 1 489 661 & 3 333 864 & 479   & 6 960 & 20 min \\
    Florida      & 1 387 005 & 3 488 194 & 655   & 5 325 & 30 min \\
    California   & 2 178 025 & 5 394 762 & 1 010 & 5 341 & 86 min \\
    Texas        & 2 884 762 & 7 194 984 & 1 336 & 5 385 & 135 min \\
    \hline
  \end{tabular*}
\end{table}

\section{The Route Planning Algorithm -- Queries}
\label{sec:query}

Having already computed the boundary graph and $\cal S$-reach in the 
preprocessing phase, we now describe a natural simplified two-stage query 
algorithm based on the former. In its {\em cellular} stage, as outlined in the 
introduction, the algorithm runs ${\cal{S}}$-Dijkstra's search until all its
branches get saturated at cell boundaries (typically, only one or two adjacent 
cells are searched). Then, in the {\em boundary} stage, virtually any 
established route planning algorithm may be used to finish the search 
(cf.~Prop.~\ref{prop:BEdistance}) since $B_{\cal E}$ is a relatively small 
graph (Table~\ref{tab:preprocessing}) and is free from scope consideration.
E.g., we use the standard reach-based A$^*$ \cite{Goldberg2005A}.

\begin{description}
\item[Two-stage Query Algorithm \rm(simplified).]
  Let a road network $(G,w)$, a proper scope mapping ${\cal{S}}$, an 
  $\cal S$-reach $r^{\cal{S}}$ on $G$, and the boundary graph $B_{\cal E}$ 
  associated with an edge partition $\cal E$ of $G$, be given. Assume start 
  and target positions $s,t \in V(G)$; then the following algorithm computes, 
  from $s$ to~$t$, an optimal $st$-admissible and saturated walk in $G$ for 
  the scope~$\cal S$.
\end{description}
\vspace*{-2ex}
\begin{enumerate}
\item \label{it:stepi}\textit{Opening cellular stage.}
  Let $I_s\subseteq G$ initially be the subgraph formed by the cell (or a 
  union of such) containing the start $s$. Let $\Gamma(I_s)\subseteq V(I_s)$ 
  denote the actual boundary of $I_s$, i.e.\ those vertices incident both 
  with edges of $I_s$ and of the complement (not the same as the union of cell 
  boundaries).
  \begin{enumerate}
  \smallskip
  \item\label{it:stepa}
    Run {\em$\cal S$-reach ${\cal{S}}$-Dijkstra's} algorithm (unidir.) on 
    $(I_s,w)$ starting from~$s$.
  \item  Let $U$ be the set of non-\,$s$-saturated vertices in $\Gamma(I_s)$
    accessed in (\ref{it:stepa}). As long as $U\not=\emptyset$,~ let $I_s\is 
    I_s\cup J_U$ where $J_U$ is the union of all cells containing some vertex 
    of~$U$, and continue with (\ref{it:stepa}).
  \end{enumerate}
  \smallskip
  An analogical procedure is run concurrently in the reverse network $(G^R,w)$ 
  on the target $t$ and $I_t$. If it happens that $I_s,I_t$ intersect and a 
  termination condition of bidirectional Dijkstra is met, then the algorithm 
  stops here.
  \smallskip
\item \label{it:stepbd}\textit{Boundary stage. } Let $B_{s,t}$ be the graph 
  created from $B_{\cal E}\cup\{s,t\}$ by adding the edges from $s$ to each 
  vertex of $\Gamma(I_s)$ and from each one of $\Gamma(I_t)$ to $t$. The
  weights of these new edges equal the distance estimates computed 
  in~(\ref{it:stepi}). (Notice that many of the weights are actually 
  $\infty$\,---can be ignored---since the vertices are inaccessible, e.g., due 
  to scope admissibility or $\cal S$-reach.)
  \par\smallskip
   Run the {\em standard reach-based A$^*$} algorithm on the weighted graph 
   $B_{s,t}$ (while the reach refers back to $B_{\cal E}$), to find an optimal 
   path $Q$ from $s$ to~$t$.
   \smallskip
 \item \textit{Closing cellular stage.} The path $Q$ computed in 
   (\ref{it:stepbd}) is easily ``unrolled'' into an optimal $st$-admissible 
   saturated walk $P$ from $s$ to $t$ in the network $(G,w)$.
\end{enumerate}

A simplification in the above algorithm lies in neglecting possible 
non-saturated walks between $s,t$ (cf.\ Theorem~\ref{thm:Sreach}), which may 
not be found by an $\cal S$-reach-based search in (\ref{it:stepi}). This 
happens only if $s,t$ are very close in a local neighborhood wrt.~$\cal S$.

We provide here a formal description of our query algorithm in pseudocode,
see Algorithm~\ref{alg:SQuery}. One minor aspect of this algorithm dealing 
with possible ``unsaturated'' optimal walks in short-distance queries which 
cannot be properly handled by $\cal S$-reach-based search, is closely 
described next. 

We also briefly comment on the use of Dijkstra's and A* algorithms in the
routing query Algorithm~\ref{alg:SQuery}: While A* is obviously much better 
in general situations, and it is used in the boundary stage, the cellular 
stage is genuinly different in the fact that the search has to get to an 
$s$-saturated state in all directions, and so A* would be of no help in the 
cellular stage (the speed-up there is achieved by different means).

One more term related to Definition~\ref{def_cell} of graph partitioning and
cells is needed in the algorithm pseudocode: For $v\in V(G)$ let $J_{\cal E}(v)$
be the set of all indices $1\leq j\leq \ell$ such that $v\in V(C_j)$.

\begin{algorithm}
\caption{~Route Planning Query Algorithm}
\label{alg:SQuery}
\begin{algorithmic}[1]   
  \smallskip
  \REQUIRE A road network $(G,w)$, a proper scope mapping ${\cal{S}}$, the
  associated $\cal{S}$-reach $r^{S}$ (or an upper bound on it), a
  partition $\cal E$ of $E(G)$ with the associated boundary graph $B_{\cal E}$ 
  (of distance weighting $p$), and start and target vertices $s,t \in V(G)$. 
  \medskip
  \ENSURE An optimal $st$-admissible walk from $s$ to $t$ in this road network 
  (or $\infty$).
\end{algorithmic}
\medskip
\underline{\textsc{Query}$(G,w,{\cal{S}},{\cal{E}},B_{\cal E},r^{\cal{S}},s,t)$} 
\vskip 3pt
\begin{algorithmic}[1]
  \STATE \textsc{Unsaturated-Local-Search}$(G,w,{\cal S},s,t)$
  \COMMENT{\hfill // Preliminary ``unsaturated'' stage.}
  \vskip 1pt
  \COMMENT{Apart from the main query algorithm below, a local search
  for possible non-saturated optimal $s$--$t$ walks has to be run first
  (in case of close positions $s,t$). Details are described separately.}
  \vskip4pt
  \STATE $I_s=\bigcup_{j\in J_{\cal E}(s)}C_j$; $I'_s \is \emptyset$
  \COMMENT{\hfill // Opening cellular stage (from the start).}
  \WHILE{$I'_s \neq I_s$}
    \STATE $I'_s \is I_s$
    \STATE $d_s,\pi_s \is {\cal{S}}\textsc{-Dijkstra}(I_s,w\rst{I_s},
    {\cal{S}}\rst{I_s},s,{r^{\cal{S}}}\rst{I_s})$
    \STATE $U \is \{ u \in \Gamma(I_s) \, | \, \forall 
    \ell\in Im({\cal{S}})\setminus\{\infty\}.~ \sigma_{\ell}[u]>
    \nu^{\cal{S}}_{\ell}\}$
    \STATE $I_s\is I_s\cup\bigcup_{j\in J_{\cal E}(U)}C_j$
  \ENDWHILE
  \STATE $R_s \is \{ v \in V(B_{\cal E}) \, | \, d_s[v] < \infty\}$
  \smallskip
  \STATE $I_t=\bigcup_{j\in J_{\cal E}(t)}C_j$; $I'_t \is \emptyset$
  \COMMENT{\hfill // Opening cellular stage (to the target).}
  \WHILE{$I'_t \neq I_t$}
    \STATE $I'_t \is I_t$
    \STATE $d_t,\pi_t \is \textsc{\cal{S}-Dijkstra}(I_t^R,w\rst{I_t},
    {\cal{S}}\rst{I_t},t, {r^{\cal{S}}[G^R]}\rst{I_t})$
    \STATE $U \is \{ u \in \Gamma(I_t) \, | \, \forall 
    \ell\in Im({\cal{S}})\setminus\{\infty\}.~ \sigma_{\ell}[u]>
    \nu^{\cal{S}}_{\ell}\}$
    \STATE $I_t\is I_t\cup\bigcup_{j\in J_{\cal E}(U)}C_j$
  \ENDWHILE
  \STATE $R_t \is \{ v \in V(B_{\cal E}) \, | \, d_t[v] < \infty\}$
  \smallskip
  \STATE $B_{s,t} \is B_{\cal E}; ~ V(B_{s,t})\is V(B_{s,t}) \cup \{s,t\}$
  \vskip 4pt
  \FORALL[\hfill // Boundary stage.]{$u \in R_s$}
    \STATE $E(B_{s,t}) \is E(B_{s,t}) \cup (s,u); ~ p(s,u) \is d_s[u]$
  \ENDFOR
  \FORALL{$u \in R_t$}
    \STATE $E(B_{s,t}) \is E(B_{s,t}) \cup (u,t); ~ p(u,t) \is d_t[u]$
  \ENDFOR \vskip 3pt
  \STATE $d_{s,t},(u_0,v_0,\ldots,u_k,v_k)\is 
  \textsc{Reach-Based-Bidirectional-A*}(B_{s,t},p,s,t)$
  \vskip 2pt
  \COMMENT{Recall that the reach values in this algorithm
  refer back to $B_{\cal{E}}$ (and not to $B_{s,t}$). Therefore, for a vertex $v$ 
  accessed from $u$, one has to remember which edge $e_s$ the optimal walk from 
  $s$ to $u$ seen so far starts within $B_{s,t}$. The weight $w(e_s)$ must be 
  subtracted from $d[u]$ while checking the reach admissibility condition.}
  \vskip 4pt
  \FORALL[\hfill // Closing cellular stage.]{$(u_i,v_i) \in E((u_0,v_0,\ldots,
    u_k,v_k))$}
    \STATE $C_i \is$ a cell containing edge $(u_i,v_i)$
    \STATE $d_i,\pi_i \is \textsc{Bidirectional-A*}(C_i^{\infty}, 
    w\rst{C_i^{\infty}} ,u_i,v_i)$
  \ENDFOR
  \STATE $P_{s,t} \is \pi_s[u_0] . \pi_0 \ldots \pi_k . \pi_t^R[v_k]$
  \smallskip
  \RETURN $d_{s,t},~ P_{s,t}$
\end{algorithmic}
\end{algorithm}

\subsubsection{Unsaturated Local Search.} 
Relating to line 1 of Algorithm~\ref{alg:SQuery},
we now explain the necessity of running a separate local search for possible
non-saturated optimal $s$--$t$ walks in the road network.
As noted in the main paper, it is an inevitable limitation of the $\cal
S$-reach $\cal S$-Dijkstra's algorithm that it may not search/discover an
$s$--$t$ walk that is optimal but not saturated (cf.\ Theorem~\ref{thm:Sreach}).
Note that the presence of such a {\em singular walk} cannot be simply tested with
a condition $I_s\cap I_t\neq \emptyset$ in the cellular stage of the
algorithm, and so a separate local search must be used.

Though, thanks to Def.~\ref{def:xadmissible} and with well distributed scope
mapping, such a singular walk may occur only if the distance between $s,t$ is quite short,
and then even not-so-elaborate route planning algorithms may be used very
efficiently.
To be precise, if a bidirectional search is used,
then each direction is searched only until $s$- or $t$-saturated vertices
are reached, which is a quite small subset of the whole network.
(If these two search spaces do not meet, then no unsaturated $s$--$t$ walk
exists.)
Furthermore, yet another natural adaptation of the classical reach
concept---an ``anti-$\cal S$-reach'', can be used to accelerate the search.

Very briefly, the {\em anti-$\cal S$-reach} of a vertex $v$ is defined as
the maximum value of $\min \{ |P^{xv}|_w,|P^{vy}|_w \}$ over all pairs of
vertices $x,y$ such that there exists a non-saturated optimal
$xy$-admissible walk $P$ from $x$ to $y$ passing through $v$,
and $P^{xv},P^{vy}$ are its subwalks between $x,v$ and between $v,y$.
Again, this parameter is easy to compute by brute force in the network.

\subsection{Experimental evaluation.} 

Practical performance of the query algorithm has been evaluated by multiple 
simulation runs on Intel Core 2 Duo 
mobile processor T6570 (2.1 GHz) with 4GB RAM. Our algorithms are implemented 
in C and compiled with gcc-4.4.1 (with -O2). We keep track of parameters 
such as the number of scanned vertices and the queue size that influence the 
amount of working memory needed, and are good indicators of suitability of 
our algorithm for the mobile platforms. 

The collected statistical data in 
Table~\ref{tab:queries}---namely the total numbers of scanned vertices and 
the maximal queue size of the search---though, reasonably well estimate also 
the expected runtime and mainly low memory demands of the same algorithm on 
a mobile navigation device. A~sample query comparison to bidirectional 
Dijkstra's algorithm is then depicted in Fig. \ref{fig:querycomparison}, and 
some statistical data are briefly summarized in Table \ref{tab:querycomparison}.

\begin{figure}[H]
  \centering
  \centerline{\epsfig{file=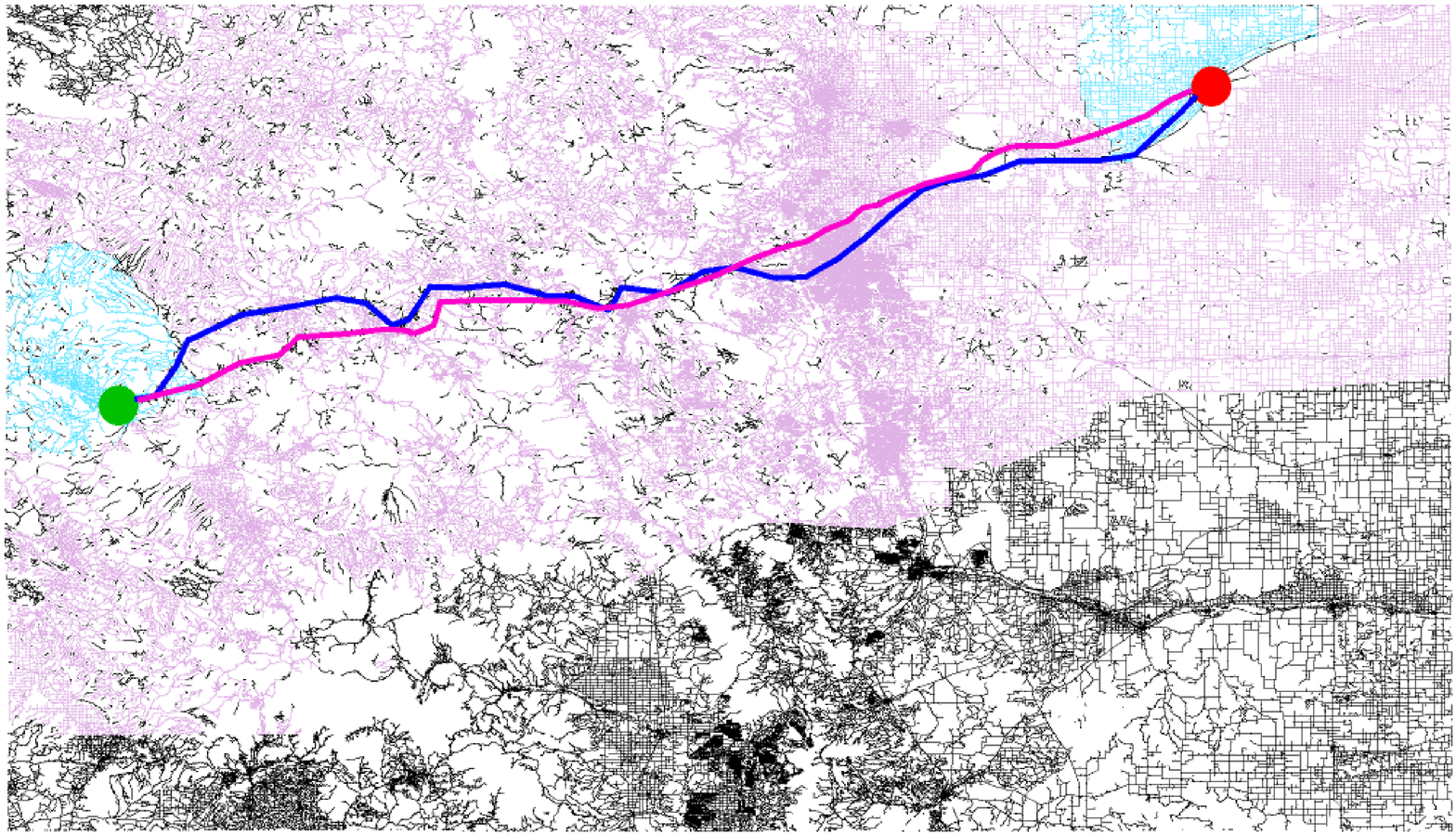, scale=.35}}
  \vspace{-3ex}
  \caption{An example of solving the routing query. The optimal $st$-admissible 
    walk from the start $s$ (green) to the target $t$ (red) is depicted in 
    blue, while the overall optimal one in red. The vertices scanned during 
    our query are in light blue, and those scanned during bidirectional 
    Dijkstra's search are in light magenta. A~brief comparison with a map 
    shows that the blue route prefers major highways, making the route more 
    comfortable (and perhaps faster in practice). Colorado,  CA, United States.}
  \label{fig:querycomparison}
\end{figure}

\begin{table}[H]
  \caption{A comparison of our approach with bidirectional Dijkstra's algorithm
    (\mbox{B-Dijkstra}). The US-all (continental part of the US) road network 
    is used. Each row lists average for 1 000 uniformly chosen pairs of start 
    and target vertices for distance intervals. It could be easily seen that 
    the number of scanned vertices increases only a slightly (more vertices of 
    the boundary graph are scanned) in our query algorithm, while the number 
    of vertices scanned during bidirectional Dijkstra's search increases 
    significantly.}
  \label{tab:querycomparison}
  \centering
  \begin{tabular}{@{\qquad}c@{\qquad}c@{\qquad}c@{\qquad}}
   \hline \hline
    \textit{Distance} & \textit{Scanned vertices} & \textit{Max. queue size} \\
    \scriptsize{(km)} & \scriptsize{our approach / B-Dijksta} & 
    \scriptsize{our approach / B-Dijksta} \\
    \hline
    2000 - 3000 & 3 225 / 28 139 196 & 60 / 4276 \\
    1000 - 2000 & 2 854 / 14 821 472 & 57 / 2477 \\
    \,~500 - 1000 & 2 305 / \,~2 147 015 & 53 / 2011 \\
    \,~~~\,0 - ~500 & 1 714 / ~~~~692 769 & 41 / 1587 \\    
    \hline
  \end{tabular}
\end{table}

\begin{table}[H]
  \caption{Experimental queries for the preprocessed continental US road 
    network (USA-all) from Table~\ref{tab:preprocessing}. Each row carries 
    statistical results for $1000$ uniformly chosen start--target pairs with 
    saturated optimal walks. The table contains average values for the numbers 
    of cells hit by the resulting optimal walk, the overall numbers of 
    vertices scanned in~$G$ during the cellular and in $B_{\cal E}$  during the 
    boundary stage, the maximal numbers of elements in the processing queue
    during the cellular and the boundary stages, and the average time spent in 
    the cellular and the boundary stages. We remark that this statistics skips 
    the closing cellular stage since the computed walk can be ``unrolled'' 
    inside each cell on-the-fly while displaying details of the route.}
  \label{tab:queries}
  \centering
  \begin{tabular*}{\textwidth}{r@{\quad}r@{\quad}c@{\quad}l@{\quad}r@{\quad}c}
    \hline \hline     
    \em~~Query distance& ~~~Hit~\,  & \em~Scanned vertices & 
	\em~~Max.\ queue size  &\em Query time \\
     (km)\qquad\, & cells~ & ~~{\scriptsize cellular / boundary} & 
	~~{\scriptsize cellular / boundary} & ~~(ms)~~~~~ \\
    \hline
3000 - ~~~$\infty$ ~ & 277~ & 1 392 / 3 490 & ~~~~~~~60 / 58 & 8.2 + 29.8 \\
2000 - 3000 ~ & 139~ & 1 411 / 1 814 & ~~~~~~~64 / 52 & 7.9 + 26.9 \\
1000 - 2000 ~ &  57~ & 1 343 / 1 511 & ~~~~~~~57 / 49 & 7.7 + 22.8 \\
500 - 1000 ~ &  25~ & 1 113 / 1 192 & ~~~~~~~53 / 38 & 8.1 + 19.0 \\
0 - \,~500 ~ &  10~ & ~~\,998 / ~~\,716 & ~~~~~~~41 / 34 & 6.9 + 16.1 \\ 
    \hline
  \end{tabular*}
\end{table}

\section{Conclusions}
\label{sec:conclusion}

We have introduced a new concept of {\em scope} in the static route planning
problem, aiming at a proper formalization of vague ``route comfort'' based on
anticipated additional metadata of the road network. At the same time we have 
shown how the scope concept {\em nicely interoperates} with other established 
tools in route planning; such as with vertex-separator partitioning and with 
the reach concept.
Moreover, our approach allows also a smooth incorporation of local route 
restrictions and traffic regulations modeled by so-called {\em maneuvers}
\cite{HM2011-maneuvers}.

On the top of formalizing desired ``comfortable routes'', the proper mixture 
of the aforementioned classical concepts with scope brings more added values;
very small size of auxiliary metadata from preprocessing 
(Table~\ref{tab:preprocessing}) and practically very efficient optimal routing
query algorithm (Table~\ref{tab:queries}). The {\em small price to be paid} 
for this route comfort, fast planning, and small size of auxiliary data 
altogether, is a marginal increase in the weight of an optimal scope 
admissible walk as compared to the overall optimal one (scope admissible walks 
form a proper subset of all walks). Simulations with the very basic scope 
mapping from Table~\ref{tab:scope_example} reveal an average increase of less 
than $3\%$ for short queries up to $500$km, and $1.5\%$ for queries above 
$3000$km. With better quality road network metadata and a more realistic 
utility weight function (such as travel time) these would presumably be even 
smaller numbers. 

\smallskip
At last we very briefly outline two directions for further research on the 
topic.\vspace{-.8ex}%
\begin{enumerate}[i.]
\item With finer-resolution road metadata, it could be useful to add a few 
  more scope levels and introduce another query stage(s) ``in the middle''.
\item The next natural step of our research is to incorporate dynamic road 
  network changes (such as live traffic info) into our approach---more 
  specifically into the definition of scope-admissible walks; e.g., by locally 
  re-allowing roads of low scope level nearby such disturbances.
\end{enumerate}

\bibliographystyle{plain}
\bibliography{references}

\end{document}